\documentclass{article}
\usepackage{graphicx} 
\usepackage[authoryear,longnamesfirst]{natbib}
\usepackage{wrapfig}
\usepackage{url}
\usepackage{amsthm,mathtools,bm}
\usepackage{natbib}
\usepackage{algorithm}
\usepackage{algpseudocode}
\usepackage{placeins}
\usepackage{booktabs}
\usepackage{multirow}
\usepackage{amsmath}

\usepackage{amssymb}
\usepackage{hyperref}
\usepackage[utf8]{inputenc}
\usepackage[T1]{fontenc}
\usepackage{tikz-cd}
\usepackage{graphicx} % in preamble
\usepackage{subcaption}
\newtheorem{assumption}{Assumption}
\usepackage[dvipsnames]{xcolor}
\tikzset{smalltext/.style={"\textup{\small #1}" description}}
\tikzstyle{startstop} = [rectangle, rounded corners, minimum width=2.5cm, minimum height=0.8cm, text centered, draw=black, fill=white!30]
\tikzstyle{process} = [rectangle, minimum width=3cm, minimum height=1.5cm, text centered, draw=black, fill=blue!30]
\tikzstyle{arrow} = [thick,->,>=stealth]
\tikzstyle{data} = [rectangle, minimum width=2.8cm, minimum height=0.8cm, text centered, draw=black, fill=green!30]

\newtheorem{problem}{Problem}%[section]
\newtheorem{theorem}{Theorem}
\newtheorem{rmk}{Remark}
\title{Non-intrusive nonlinear reduced-order modeling with variable projection}

\author{Dimitrios Xylogiannis\textsuperscript{\dag}\thanks{Email: dimitrios.xylogiannis@onera.fr} \and Charles Poussot-Vassal\textsuperscript{\dag}\thanks{Email: charles.poussot-vassal@onera.fr} \and Claire Sarrat\textsuperscript{\dag}\thanks{Email: claire.sarrat@onera.fr} \\[1.5ex] \textsuperscript{\dag}ONERA - The French Aerospace Lab, Toulouse, France }

\date{} 

\begin{document}

\maketitle

% Here goes the abstract
\begin{abstract}
This work presents a method for constructing nonlinear reduced-order models from input-output time-domain data. The proposed approach, termed Mixed Interpolatory Inference with Variable Projection (MIIvp), exploits the fact that the considered class of nonlinear state-space models is linear in the output equation parameters. By applying the Variable Projection (VarPro) algorithm, the optimization is restricted to the state equation parameters alone, while the output equation parameters are recovered via linear least squares. As a consequence, the output dimension does not enter the nonlinear optimization parameter vector, making the method well suited for systems with very high-dimensional outputs, a setting where many other approaches become computationally  prohibitive. Under mild assumptions, it is shown 
that MIIvp can recover the true model parameters up to similarity. The method is first validated on a synthetic bilinear system, where it achieves machine-precision accuracy and recovers the true eigenvalues. MIIvp is then compared with existing methods on two experimental benchmarks from the nonlinear system identification literature. These numerical experiments showcase both the validity and the  limitations of the proposed approach. Finally, directions for improvements and future work are outlined.
%\nocite{*}%% Remove this line from your manuscript.
\end{abstract}
\vspace{0.5em}


\vspace{0.5em}
\hrule
\vspace{1em}
\section{Introduction}\label{sec: Intro}

\subsection{Foreword}

Advances across the sciences increasingly hinge on our ability to predict, interpret, and control complex systems from limited information. In many areas of physics, chemistry, biology, and engineering, the governing laws are known in principle, yet their practical use is constrained by the cost of repeatedly running high-fidelity simulations or conducting experiments. As a result, there is a growing demand for models that can deliver fast, reliable predictions and support tasks such as forecasting, uncertainty quantification, and real-time decision-making.

Within this context, non-intrusive data-driven modeling has emerged as a particularly attractive domain; see e.g. \citet{benner2021ROM_vol1}. Rather than requiring direct access to governing equations or modifications to existing solvers, non-intrusive approaches learn predictive models by leveraging experimental measurements. This makes them broadly deployable across disciplines and compatible with established scientific workflows. At the same time, these models should not just be able to reconstruct the training data but ideally they should also generalize across different operating regimes, respect physical structure (e.g., stability, conservation of energy, symmetries), and remain interpretable enough to enable insight into the physical process under study.

A main challenge arises from the fact that many systems are only accessible through indirect measurements and the internal state that drives dynamics is either partially or not observed at all. In many cases, even if the state trajectories are measured, it might become  infeasible to store them in practice (e.g., in high-dimensional systems such as fluid flows, electromagnetic, or climate models).  In such cases, the central objective is to build models only from the available inputs and observable outputs that can explain the internal (unobserved) dynamics.  Learning models in a robust, physically consistent, and computationally efficient way is a key goal of non-intrusive modeling. These models should be of reduced complexity in order to enable simulation and analysis within seconds.

However, a further challenge arises when the output measurements are high-dimensional. This can occur in many scientific fields  where measurements are available over both space and time. For instance, in fluid and transport problems the input may be boundary conditions or forcing, while the output is a predicted spatiotemporal field such as 
velocity, pressure, or concentration. In cases like this, the internal state (the complete configuration evolving over time) is never measured in full and only sampled snapshots or sensor-derived reconstructions are available. 
Similarly, in structural dynamics the input may be an excitation waveform, while the output is a dense vibration response measured across many sensors or as a full-field displacement map, even though internal stress and damage states remain hidden. In imaging-based experiments, the output may be a high-resolution image or snapshot whose dimension grows with the number of pixels, while the underlying physical state that produces the image remains 
unobserved. Many more applications may be found in fields such as chemistry, biology, chemical and mechanical engineering.

\subsection{Problem statement}
\label{sec:Problem description}

Consider the following general nonlinear structure of 
discrete-time models with initial condition $x_0$:
\begin{equation}\label{struct: General_nonlinear}
    \left\{
	\begin{array} {rcl}
		x_{k+1}&=&Ax_{k}+Bu_{k}+E_{x}f(x_k,u_k)\\
		y_{k}&=&Cx_{k}+Du_{k}+E_{y}g(x_k,u_k)
	\end{array}
    \right. ,
\end{equation}
where $A \in \mathbb{R}^{n \times n}$, 
$C \in \mathbb{R}^{n_y \times n}$, 
$B \in \mathbb{R}^{n}$, and 
$D \in \mathbb{R}^{n_y}$ are constants. The  state, input and output variables are described by $x_{k} \in \mathbb{R}^n$, 
$u_{k} \in \mathbb{R}$ and 
$y_{k} \in \mathbb{R}^{n_y}$, respectively, for 
$k=0,1,\cdots, m-1$, where $m$ denotes the number of samples. Depending on the desired structure, 
$E_x \in \mathbb{R}^{n \times n_{f}}$, 
$E_y \in \mathbb{R}^{n_y \times n_g}$ and 
$f(x_k,u_k) \in \mathbb{R}^{n_f}$, 
$g(x_k,u_k) \in \mathbb{R}^{n_g}$ may differ depending on which terms they consider.  
In this work, the functions $f(x_k,u_k)$ and $g(x_k,u_k)$ are chosen so that a similarity transformation between equivalent realizations  exists. More specifically, we consider the following structures of \eqref{struct: General_nonlinear}:
\begin{itemize}
    \item[(L)] Linear: $E_x=E_y=0$.
    \item[(B)] Bilinear: $f(x_k,u_k)=g(x_k,u_k)= x_k u_k$, $E_x=N \in \mathbb{R}^{n \times n }$ and $E_y=F \in \mathbb{R}^{n_y \times n}$. 
    \item[(QB)] Quadratic-bilinear: $f(x_k,u_k)=g(x_k,u_k)=\begin{bmatrix} x_k u_k \\ x_k \otimes x_k \end{bmatrix}$,
$E_x=\begin{bmatrix} N & Q \end{bmatrix}$, $E_y=\begin{bmatrix} F & G \end{bmatrix}$, 

where $Q \in \mathbb{R}^{n \times n^2}$, 
$G \in \mathbb{R}^{n_y \times n^{2}}$ and $\otimes$ 
denotes the Kronecker product.
\end{itemize}

If the complexity of the model is high (i.e.\ $n$ is large), Model Order Reduction (MOR) provides appropriate methods to derive a simpler model of order $r \ll n$. If the matrices are known, MOR methods are referred to as  \emph{intrusive}. Conversely, when only data are available,  methods are referred to as \emph{non-intrusive} (see details in \cite{antoulas2005approximation} and \cite{AntoulasBook:2020})\footnote{Notice that the terms non-intrusive and identification methods are closely related. The difference mostly comes from the associated research community: the non-intrusive term comes from the model reduction and scientific computation communities, while system identification come from the dynamical systems community; both adopt a different methodology but share a similar objective; \cite{antoulas2005approximation} and \cite{Schoukens:2016}.}. 

Our main goal is to construct models of the form \eqref{struct: General_nonlinear} based \emph{solely} on time-domain input-output data  $\{u_k, y_k\}_{k=0}^{m-1}$. %The reduced order modeling and system  identification communities have considered many methods to address this challenge. 
The central problem addressed in this work can be stated as follows:
\begin{problem}\label{problem}
Given input-output data $\{u_k, y_k\}_{k=0}^{m-1}$ with $u_k \in \mathbb{R}$ and $y_k \in \mathbb{R}^{n_y}$, with $n_y$ large, construct a nonlinear dynamical model of the form \eqref{struct: General_nonlinear} with state dimension $r \ll n$, that captures the input-output behavior represented by the data.

\end{problem}

\subsection{Related work}
\label{sec:biblo}
The literature on nonlinear reduced-order modeling and system identification is vast, spanning a wide range of tools and methodologies from different scientific fields.
 A commonly applied technique is to solve a nonlinear optimization problem between the measured output and the predicted one. In this case, the optimization variable gathers all the vectorized matrix parameters of the desired system structure.
%\begin{equation}{\label{nonlinlsq}}
 %   \min_{\theta} \sum_{k=0}^{m-1} 
  %  \|y_{k} - \hat{y}_{k}(\theta)\|^{2}_{2},
%\end{equation}
%where $\theta$ gathers all the vectorized matrix parameters. In particular, for \eqref{struct: General_nonlinear} it is defined as:

%\begin{equation}
 %   \theta\coloneqq\; 
%\bigl(\mathrm{vec}(A)^\top, 
%B^\top, \mathrm{vec}(E_x)^\top, \mathrm{vec}(C)^\top, D^\top,\mathrm{vec}(E_y)^\top  \bigr)^\top,
%\end{equation}
%where  $\mathrm{vec}$  denotes the vectorization operator. One may also add the initial condition $x_0$ of the model as an extra parameter to be optimized.

Since this optimization problem  is nonlinear, no analytical solution exists and one has to rely on nonlinear optimization. Two of the methods that consider this approach and have been successfully applied in many applications are  Polynomial Nonlinear State-Space (PNLSS) \citep{PADUART2010_pnlss} and the related   Nonlinear State-Space (NLSS) \citep{Marconato2014ImprovedIF}. These methods first rely on  the Best Linear Approximation (BLA) from \cite{Pintelon_Schoukens_book2012}. Then, one solves a nonlinear problem  by  initializing with this linear model and setting the nonlinear matrices as zero. However, one main issue is the large dimension of the optimized parameters. This may result in a substantial computational time to construct the model  and therefore makes the  method impractical. Remedies for this have been  suggested in \cite{Decuyper_decouple}. However, if the output dimension $n_y$ is very large, even for low order $n$, the number of the parameters grows fast.  Other successful nonlinear identification methods use neural networks  such as \cite{FORGIONE2021, Forgione2020dynoNet}.  Promising nonlinear methods that can deal with  high-dimensional output series are presented  by \cite{WAHLSTROM2015} and Subspace Encoder Approach (SUBNET)  \citep{Beintema21a, BEINTEMA2021_video, BEINTEMA2023} 
that is based on Artificial Neural Networks (ANNs). Authors suggest the Mixed Interpolatory Inference (MII)  in  \cite{MixedII} and \cite{MyCDC}, and an extra step to tackle the case of noisy data was added in \cite{MyECC}. Furthermore, in \cite{MyECC} a nonlinear optimization problem, as  PNLSS, was considered to further improve the data fit.  However, as already mentioned, this nonlinear optimization step cannot be efficiently applied in the case of a high-dimensional output. In this case, the dimension of the optimization problem might grow very large  and solving it would become practically infeasible. This work aims to replace the nonlinear optimization with the Variable Projection (VarPro) \citep{Golub1973}, in order to reduce the number of parameters.

VarPro has been previously considered for model 
discovery in related contexts. Relatively recently, VarPro was used in \cite{OptDMD2018} to modify the well-known method of Dynamic Mode Decomposition (DMD), that can be applied when state trajectories are available. In \cite{BRULS1999}, the authors used 
VarPro for linear identification and also for 
identification of a Wiener model. However, for the Wiener model, the nonlinear parameters were all the matrices of the linear part. In~\cite{Verdult2001}, the authors considered separable least squares optimization to build bilinear models, but due to their choice of parametrization, the output dimension was included during  optimization. VarPro was also considered in~\cite{GanMin2021} to construct radial-basis-function network-based autoregressive models with exogenous input (RBF-ARX) models.

\subsection{Contribution}
\label{sec:contribution}

This work aims at addressing Problem \ref{problem}, and particularly the issues described above by proposing a method that can construct nonlinear models \eqref{struct: General_nonlinear} even with a very large output dimension. This is achieved by  observing that although the models \eqref{struct: General_nonlinear} are nonlinear, they are \emph{linear in the parameters}.

Therefore, the nonlinear optimization problem over all parameters can be reformulated as a separable nonlinear least squares problem and solved efficiently with the VarPro  algorithm. This boils down to solving a nonlinear optimization problem that only considers the parameters of the state equation. Then, the parameters of the output equation can be obtained by solving a linear least squares problem. Hence, the number of nonlinear optimization variables is independent of the output dimension, making it feasible to use this method for large output dimensions.

The proposed method, termed Mixed Interpolatory  Inference with Variable Projection (MIIvp), extends the MII framework with a VarPro refinement step. The main contributions of this work are:

\begin{enumerate}
\item A non-intrusive method for constructing 
nonlinear models~\eqref{struct: General_nonlinear} 
from input-output data whose optimization problem dimension is independent of the output dimension~$n_y$, while remaining applicable to data generated by  large-scale dynamical systems.
\item A local convergence result 
showing that, under assumptions, MIIvp recovers the true model parameters up to similarity.

\end{enumerate}

The proposed framework was partly motivated by the authors' prior hands-on practice in atmospheric dispersion modeling, where the input is the flow rate of the pollutant and the output is the high-dimensional field of the plume concentration over space and time \citep{Sarrat2017airport,MixedII}. Two main challenges arise with the 
pollutant datasets: (\textit{i}) the high-dimensional output field and (\textit{ii}) the low number of samples available due to costly Large Eddy Simulations (LES). However, these results are beyond the scope of the present article and will be reported in a dedicated future study.

The remainder of this work is organized as follows: 
Section \ref{sec:MII} reviews the steps of MII. The 
proposed framework (MIIvp) is described in 
Section \ref{sec:MIIvp}. The performance of MIIvp is 
tested in Section \ref{sec: Numerics} on a simple academic model 
and two challenging nonlinear benchmarks, where comparison with other existing approaches is possible. Finally,  Section \ref{sec: Conclusion} gives the conclusions and suggests directions for future work.

\section{Mixed Interpolatory Inference (MII)} \label{sec:MII}

MII basically consists of three steps: the Pencil method \citep{IonitaPencil:2012}, the Loewner framework \citep{Loewner}  and an inference step \citep{GoseaDMD}; see \cite{MyCDC} for details.  For simplicity, we review each step for the Single-Input Single-Output (SISO) case (i.e. $n_u=n_y=1$) but each step can be easily extended for the Multi-Input Multi-Output (MIMO) case. 

\subsection{Pencil method}

The Pencil method  provides a realization for the linear case, as proposed in \cite{IonitaPencil:2012} and is closely related to the \cite{Ho_Kalman1966} method. Let us suppose that we have an input-output pair $\{u_{k},y_{k}\}$  for $k=0,1,\cdots,m-1$ with $m\in\mathbb{N}$. In order to calculate the Markov parameters denoted by $h \in \mathbb{R}^{m}$ we form  an input  matrix $\mathcal{U}$ as:
\begin{equation}
    \mathcal{U}=\begin{pmatrix}
    u_0 & \\
          u_{1} &  u_{0}\\
 \vdots &  \ddots  &  \ddots\\
  u_{m-1} & \cdots &u_1& u_{0}\\
\end{pmatrix}\in \mathbb{R}^{m \times m},
\end{equation}
and solve the following least squares problem 
\begin{equation}{\label{OLS_pencil}}
    \hat{h}:=\arg \min_h \frac{1}{2}\|\mathcal{U}h - y ^{\top}\|_2^2 .
\end{equation}
After obtaining the Markov parameters $h=(h_{0},\cdots,h_{m-1})^\top$, the Hankel matrix $\mathcal{H}$ can be formed as follows:
\begin{equation*}
\mathcal{H}=  \begin{pmatrix}
        h_{1} &  h_{2}  &  h_{3}& \cdots&  h_{p+1}\\
          h_{2} &  h_{3}  &  h_{4}& \cdots & h_{p+2}\\
 \vdots &  \vdots  &  \vdots&  \cdots& \vdots\\
  %h_{p-1} &  h_{p-2}  &  h_{p-3}& \cdots & h_{2p-1}\\
   h_{p} &  h_{p+1}  &  h_{p+2}& \cdots & h_{2p-1}    
    \end{pmatrix}.
\end{equation*}
 For simplicity, we assume that $m=2p$, where $p \in \mathbb{N}$. For zero initial conditions, a descriptor realization $\bigl($i.e. $Ex_{k+1}=Ax_{k}+Bu_{k}, y_{k}=Cx_{k}+Du_{k} \bigr)$ with $E \in \mathbb{R}^{p \times p}$ and $D \in \mathbb{R}$, is given by:
\begin{align*}
    &E=\mathcal{H}(1:p,1:p),\, A=\mathcal{H}(1:p,2:p+1),\\ 
    &B=\mathcal{H}(1:p,1),\,  C=\mathcal{H}(1,1:p),\, D=h_{0}.
\end{align*}

By construction, the above realization is not necessarily minimal. However, applying  a rank-revealing factorization and projection on the $[E,A]$ matrix yields a minimal realization. Compared to  Eigensystem Realization Algorithm (ERA) \citep{ERA1985}, which assumes a unit impulse input, the Pencil method can be applied for an arbitrary input and therefore can be viewed as a generalization of ERA.

\subsection{Model reduction step}

Depending on the data and the nature of the application, the intermediate linear model may be of large order. To have a model with reduced complexity $r$ $(r \ll p)$, model reduction techniques can be used. Among the various  techniques, we suggest the use of the Loewner framework (LF) by  \cite{Loewner}, due to its data-driven nature and flexibility.  The details of this method can also be found in the tutorial  \cite{AntoulasSurvey:2016}. After MOR, we end up with a linear reduced-order model  as follows:

\begin{equation}
\left\{
	\begin{array}{rcl}
        E^{r}x^{r}_{k+1}&=&A^{r}x^{r}_{k}+B^{r}u_{k}, \hspace{0.2cm}    \hspace{0.4cm} x^{r}_{0}=0,\\
		y^{r}_{k}&=&C^{r}x^{r}_{k}, % \hspace{0.2cm} 
	\end{array}
 \right.
\end{equation}
where  $E^{r},A^{r} \in \mathbb{R}^{r \times r}$, $(C^{r})^{\top}, B^{r} \in \mathbb{R}^{r}$, $x^{r}_{k} \in \mathbb{R}^r$, $u_{k} \in \mathbb{R}$ and  $ y^{r}_{k} \in \mathbb{R}$, where $r\ll p$. Notice that the $D$ term vanishes, and is incorporate into the $E^r$ matrix. 

Now  the reduced state $x^r$  will be combined with  the original input-output data to add the nonlinear term, as the next section shows.  Notice that for the first step of the method any linear system identification method may be used. 
The combination of the Pencil method and the Loewner framework is motivated by (i) our data-driven modeling objective in a realistic setting, and by (ii) the fact that they both scale with a large number of internal variables $n$. The Loewner framework relies solely on transfer-function evaluations at prescribed frequencies; therefore, the linear model identified by the Pencil method, which may be of large dimension, need not be explicitly stored. This leads to a substantial reduction in memory usage and computational burden in practical setups; thus addressing some "curse of dimensionality" issues \citep{AGPV25}.  Note that it is well-known  that neither the Pencil method nor the Loewner framework guarantees that the derived linear model will be stable. In this work, if the linear system is unstable, we apply a projection onto the $\mathcal{RH}_\infty$ space, as described in \cite{Enforcing_stability}.  

\subsection{Inference step}

This step is a time-domain method, as opposed to the Loewner framework which applies in the frequency-domain \citep{GoseaDMD}. 
After the model reduction step, the  state $x^{r}_k \in \mathbb{R}^{r}$ of the linear reduced order model  is collected. Then,  the original input-output data and the   state snapshots are sorted  as follows:

\begin{equation}
\begin{array}{rcl}
    U&:=&\begin{bmatrix}
        u_{0} &  u_{1} & \cdots &  u_{m-1}
    \end{bmatrix} \in \mathbb{R}^{1 \times m}\\
    Y&:=&\begin{bmatrix}
        y_{0} &  y_{1} & \cdots &  y_{m-1}
    \end{bmatrix} \in \mathbb{R}^{1\times m}\\
    X^{r}&:=&\begin{bmatrix}
        x^{r}_{0} &  x^{r}_{1} & \cdots &  x^{r}_{m-1}
    \end{bmatrix} \in \mathbb{R}^{r \times m}\\
    X^{r}_{s}&:=&\begin{bmatrix}
        x^{r}_{1} &  x^{r}_{2} & \cdots &  x^{r}_{m}
    \end{bmatrix} \in \mathbb{R}^{r \times m}.
\end{array}
\end{equation}

Define the matrix $U_{d}=diag(u_{0},u_{1},\cdots,u_{m-1})\in \mathbb{R}^{m \times m}$. For both the bilinear and quadratic bilinear structure of \eqref{struct: General_nonlinear}, one can solve the following  problem:
\begin{equation}{\label{lsq}}
    \min_{\hat{A},\hat{B},\hat{C},\hat{D},\hat{E}_x,\hat{E}_y}   \dfrac{1}{2}\Biggl\|\begin{bmatrix}
        X^{r}_{s}\\ Y
    \end{bmatrix} -\begin{bmatrix}
        \hat{A} &\hat{B}&\hat{E}_x\\ \hat{C} &\hat{D}&\hat{E}_y
    \end{bmatrix}\begin{bmatrix}
        X^{r}\\ U \\ Z
    \end{bmatrix}\Biggr\|^{2}_{F},
\end{equation}
where $\|.\|_{F}$ denotes the Frobenius norm and the matrix $Z \in \mathbb{R}^{z \times m}$ is defined as  $Z=X^{r}U_{d}$ for the bilinear structure of \eqref{struct: General_nonlinear} and $Z=\begin{bmatrix}
    X^{r}U_{d} \\ X^{r} \odot X^{r}
\end{bmatrix}$ for the quadratic-bilinear case of \eqref{struct: General_nonlinear}, where $\odot$ denotes the Khatri-Rao product. 
By defining the matrices:
\begin{equation}
\begin{array}{rcl}
\Gamma&:=&\begin{bmatrix}
        X^{r}_{s}\\ Y
    \end{bmatrix}\in \mathbb{R}^{(r+1) \times m},\\ \Xi&:=&\begin{bmatrix}
        \hat{A} &\hat{B}&\hat{E}_x\\ \hat{C} &\hat{D}&\hat{E}_y
    \end{bmatrix}\in \mathbb{R}^{(r+1) \times (r+z+1)},\\ 
   \Phi&:=&\begin{bmatrix}
        X^{r}\\ U \\Z
    \end{bmatrix} \in \mathbb{R}^{(r+z+1) \times m},
\end{array}
\end{equation}
the problem \eqref{lsq} can be written as:
\begin{equation}{\label{lsqw}}
    \min_{\Xi}  \dfrac{1}{2} \bigl\|\Gamma -\Xi\Phi\bigr\|^{2 }_{F},
\end{equation}
whose solution  is given by:
\begin{equation}
    \Xi=\Gamma \Phi^{\dag},
\end{equation}
where $\Phi^{\dag}$ denotes the Moore-Penrose inverse of matrix $\Phi$.

To avoid  ill-conditioning, Tikhonov regularization can be used. For a parameter $\mu >0$, the following problem is solved instead
\begin{equation}{\label{lsqcompact}}
    \min_{\Xi}  \dfrac{1}{2} \bigl\|\Gamma -\Xi\Phi\bigr\|^{2}_{F}+ \dfrac{\mu}{2} \bigl\|\Xi\bigr\|^{2}_{F},
\end{equation}
whose explicit solution is
\begin{equation}
    \Xi= \Gamma\Phi^{\top}(\Phi \Phi^{\top}+\mu I) ^{-1},
\end{equation}
where $I$ denotes the identity matrix. In this work, Tikhonov regularization is applied only to the bilinear structure, with the parameter $\mu$ selected by the L-curve criterion~(\cite{Hansen1992}). For the quadratic-bilinear case, because of the different scaling of the quadratic entries compared to the linear and bilinear terms,  separate regularization parameters may be required for the different blocks. This aspect, requiring deep insight, is left for future investigation.
 
 %The regularization parameter $\mu$ significantly affects the results and an appropriate selection strategy  is required. This work makes the choice by the L-curve criterion

\section{Mixed Interpolatory Inference with Variable Projection (MIIvp)} \label{sec:MIIvp}
This section presents the formulation of the problem for VarPro. The main idea is to restrict the nonlinear optimization to the parameters of the state equation alone.

\subsection{Variable Projection (VarPro)}

We consider in this section the general nonlinear form \eqref{struct: General_nonlinear}. The expression of the output equation may be rewritten in terms of the  parameters $\theta_n \in \mathbb{R}^{\tilde{n}}$ and $\theta_l    \in \mathbb{R}^{n_y \times n_{\phi}}$  defined as 
\begin{align}
\theta_n &\;\coloneqq\; 
\bigl(x_0^\top,\, \mathrm{vec}(A)^\top,\, 
B^\top,\, \mathrm{vec}(E_x)^\top\bigr)^\top 
\in \mathbb{R}^{\tilde{n}}, 
\label{eq:theta_n_def}\\
\theta_l &\;\coloneqq\; 
\begin{bmatrix} C & D & E_y \end{bmatrix} 
\in \mathbb{R}^{n_y \times n_{\phi}},
\label{eq:theta_l_def}
\end{align}
with $n_{\phi} \coloneqq n + 1 + n_g$. Moreover, define the vector 
\begin{equation}
    \label{eq: Phi}
\phi_k(\theta_n) \;\coloneqq\;
\begin{bmatrix}
x_k(\theta_n)\\[2pt]
u_k\\[2pt]
g(x_k(\theta_n),u_k)
\end{bmatrix}
\in \mathbb{R}^{n_{\phi}},
\end{equation}
and the stacked matrix
\begin{equation}
\label{eq:Psi_def}
\Psi(\theta_n) \;\coloneqq\;
\bigl[\phi_0(\theta_n)\ \phi_1(\theta_n)\ \cdots\ \phi_{m-1}(\theta_n)\bigr]
\in   \mathbb{R}^{n_{\phi}\times m}.
\end{equation}
Then, the cost function $J(\theta_n,\theta_l)$ can be defined as 
\begin{equation}
\label{eq:full_cost_vp}
J(\theta_n,\theta_l)\;\coloneqq\; \dfrac{1}{2}\bigl\|Y-\theta_l\Psi(\theta_n)\bigr\|_F^2.
\end{equation}

The above problem can be solved in two steps. 
In the first step, we seek $\hat{\theta}_n$ that minimizes the cost function
%\begin{equation}
%\label{eq:cost_vp_nonlin}
%\min_{\theta_n}\; \tilde{J}(\theta_n)\coloneqq\;\hat{\theta}^{*}_n,
%\end{equation)
\begin{equation}\label{eq:cost_vp_nonlin}
    \hat{\theta}_n \in \arg \min_{\theta_n}\; \tilde{J}(\theta_n)
\end{equation}
where 
\begin{equation}
\tilde{J}(\theta_n)\;\coloneqq\;  \dfrac{1}{2}\bigl\|Y-Y\Psi(\theta_n)^{\dagger}\Psi(\theta_n)\bigr\|_F^2.
\end{equation}

In the second step, for fixed $\theta_n =\hat{\theta}_n$,   $\theta_l$  can be computed from $Y\Psi(\hat{\theta}_n)^{\dagger}$. As before, for the linear least squares problem we may opt to add a Tikhonov regularization term.  To solve \eqref{eq:cost_vp_nonlin}, the parameter $\theta_n$ is initialized by the matrices obtained by MII and the stopping criteria can be set to be the relative change of the parameters. The  procedure of MIIvp is summarized in Algorithm~\ref{alg:MIIvp}. The optimization problem   \eqref{eq:cost_vp_nonlin} is nonlinear in the parameters and therefore a nonlinear optimization routine should be applied. Depending on the application and the available data, any nonlinear optimization solver can be used. Notice also that the gradients of the optimization problem can be computed analytically.  In this work, we choose the Limited-memory Broyden-Fletcher-Goldfarb-Shanno (L-BFGS) \citep{Liu1989}. This choice was inspired by the recent work of \cite{Bemporad2025}  to keep the framework computationally efficient. The L-BFGS algorithm estimates the inverse Hessian matrix of the optimization problem and stores only a few vectors that represent the approximation implicitly. Further considerations on the computational costs for every step of the  framework are provided in Section \ref{sec: Computational costs}.
\begin{algorithm}[h]
\caption{Mixed Interpolatory Inference with Variable Projection (MIIvp)}\label{alg:MIIvp}
\begin{algorithmic}[1]
\Require  Input-output data $\{u_k, y_k\}_{k=0}^{m-1}$  in the time-domain;  a chosen  structure of \eqref{struct: General_nonlinear}.
\Statex
\Statex \textbf{Step 1: MII}
\State  Pencil method to obtain a linear model of order $p$.
\State Apply the Loewner framework  to reduce the dimension to $r \ll p$.
\State Nonlinear inference step.
\Statex
\Statex \textbf{Step 2: VarPro}
\State Initialize with the MII model from Step 1.
\While{stopping criteria are not satisfied}
    \State Update $\theta_n$ by solving the nonlinear optimization problem \eqref{eq:cost_vp_nonlin}.
    \State Obtain $\theta_l$ by $\theta_l = Y \Psi(\theta_n)^{\dagger}$.
\EndWhile
\Statex
\State \textbf{Output:} Nonlinear reduced-order model of the form \eqref{struct: General_nonlinear}.% Add your outputs here
\end{algorithmic}
\end{algorithm}

\subsection{Local convergence of the method} \label{sec:Converegence}

Based on the original paper of  \cite{Golub1973}, we aim to show that under certain assumptions the method can converge to the true parameters.

\begin{theorem}{\cite[Theorem 2.1]{Golub1973}}

Assume that $\Psi(\theta_{n})$ has a constant rank over an open set $ \Omega \subset \mathbb{R}^{\tilde{n}}$. Then

\begin{enumerate}
    \item If $\hat{\theta}_{n}   \in  \Omega$ is a minimizer of \eqref{eq:cost_vp_nonlin} and $\hat{\theta}_l=Y\Psi(\hat{\theta}_n)^{\dagger}$, then $( \hat{\theta}_{n}, \hat{\theta}_l )$  is also a minimizer of \eqref{eq:full_cost_vp}.
    \item   If $( \hat{\theta}_{n}, \hat{\theta}_l )$ is a minimizer of \eqref{eq:full_cost_vp} for $\theta_n \in  \Omega$, then $\hat{\theta}_n$ is a minimizer of \eqref{eq:cost_vp_nonlin} in $\Omega$  and $\tilde{J}(\hat{\theta}_n)=J(\hat{\theta}_n,\hat{\theta}_l)$. Furthermore, if there is a unique $\theta_l$ among $J(\theta_n,\theta_l)$, then $\theta_l$  must satisfy $\theta_l=Y\Psi(\theta_n)^{\dagger}$.
\end{enumerate}
\label{theo:Golub_theorem}
\end{theorem}

Let us now consider the following Assumptions \ref{ass:A1}, \ref{ass:A2}, \ref{ass:A3} and \ref{ass:A4}.
\begin{assumption}
\label{ass:A1}
The input $u$ is persistently exciting and, for every $\theta_n \in \Omega$ the 
regressor matrix 
$\Psi(\theta_n) \in \mathbb{R}^{n_\phi \times m}$ 
has full row rank,
\begin{equation}
\label{eq:full_row_rank}
\mathrm{rank}\,\Psi(\theta_n) = n_\phi 
\qquad \text{with } n_\phi \le m.
\end{equation}
\end{assumption}

\begin{assumption}
\label{ass:A2}
For every $\theta_n \in \Omega$ the trajectory 
$\{x_k(\theta_n)\}_{k=0}^{m-1}$ is well-defined.
\end{assumption}

\begin{assumption}
\label{ass:A3}
If $\theta_n \in \Omega$ and there exists $\theta_l$ 
such that
\begin{equation}
\label{eq:zero_residual_representation}
Y = \theta_l \Psi(\theta_n),
\end{equation}
then there exists a similarity transformation mapping 
the true model $(\theta_n^\star, \theta_l^\star)$ to 
$(\theta_n, \theta_l)$. We write this as 
$(\theta_n, \theta_l) \sim 
(\theta_n^\star, \theta_l^\star)$.
\end{assumption}

\begin{assumption}
\label{ass:A4}
There exists an open neighborhood $\mathcal N\subset\Omega$, with $\theta_n^\star\in\mathcal N$, such that, when initialized at  $\theta_n^{(0)}\in\mathcal N$, the nonlinear solver generates a sequence $\{\theta_n^{(j)}\}_{j\geq 0}\subset\mathcal N$ satisfying
\[
\theta_n^{(j)}\to \hat{\theta}_n
\quad \text{as } j\to\infty,
\]
where
\[
\hat{\theta}_n
\in
\arg \min_{\theta_n\in\mathcal N}
\tilde J(\theta_n).
\]

\end{assumption}

Based on Theorem \ref{theo:Golub_theorem} and the assumptions above, for input-output data derived from a model of the form \eqref{struct: General_nonlinear}, the following theorem holds.

\begin{theorem}
\label{thm:MII_vp convergence}
Assume the input-output data are noiseless and are produced from a model of the form \eqref{struct: General_nonlinear}. Denote the nonlinear and its linear blocks as  $(\theta_n^\star,\theta_l^\star)$ such that
\begin{equation}
\label{eq:noiseless_fit}
Y=\theta_l^\star \Psi(\theta_n^\star).
\end{equation}
Let $\Omega \subset \mathbb{R}^{\tilde n}$ be an open set with  $\theta_n^\star \in \Omega$. Then, for $\hat\theta_n$  and the  minimizer     $\hat\theta_l \;=\; Y\Psi(\hat\theta_n)^\dagger$ of the least squares problem, it holds that $(\hat\theta_n,\hat\theta_l)\sim(\theta_n^\star,\theta_l^\star)$.
\end{theorem}

\begin{proof}
    By Assumption \ref{ass:A2}, the regressor matrix $\Psi(\theta_n)$ and the cost functions $J(\theta_n,\theta_l)$, $\tilde{J}(\theta_n)$  are well-defined on $\Omega$.

    Let $\mathcal N  \subset \Omega$ with $\theta^{\star}_n \in \mathcal N$.
    For the solution of \eqref{eq:cost_vp_nonlin} $\hat{\theta}_n \in \mathcal N $ with $\hat\theta_n\in\arg\min_{\theta_n\in\mathcal N}\tilde J(\theta_n)$, from Assumption \ref{ass:A4} it holds that $\tilde J(\hat\theta_n)\;\le\;\tilde J(\theta_n^\star)=0,$ and since $\tilde J(\hat\theta_n)\ge 0$, we conclude that 
$\tilde J(\hat\theta_n)=0$.

    By Assumption \ref{ass:A1}, $\Psi(\theta_n)$ has full row rank on $\Omega$. Therefore, by Theorem \ref{theo:Golub_theorem}, for fixed $\hat{\theta}_n$, $\hat{\theta_l}= Y \Psi(\hat{\theta}_n)^{\dagger}$ is a least squares minimizer, which is also unique by Assumption \ref{ass:A1}, and by Theorem \ref{theo:Golub_theorem} it satisfies 
$J(\hat{\theta}_n,\hat{\theta}_l)=\tilde{J}(\hat{\theta}_n)$.

Since  $\tilde{J}(\hat{\theta}_n)=0$, this yields $J(\hat{\theta}_n,\hat{\theta}_l)=0$
and therefore, $Y=\hat{\theta}_l\Psi(\hat{\theta}_n)$. By Assumption \ref{ass:A3}, the above yields   $(\hat{\theta}_n,\hat{\theta}_l)\sim (\theta_n^\star,\theta_l^\star)$. This completes the proof.

\end{proof}

\begin{rmk}
Assumption \ref{ass:A4} is a local convergence 
requirement. It asserts the existence of a basin of attraction  within which the nonlinear solver converges to a minimizer of $\tilde{J}$. In practice, the cost function$\tilde J(\theta_n)$ is nonconvex. Consequently, nonlinear solvers  typically only guarantee to converge to a stationary point and may converge to local minima or saddle points, depending on the initialization and numerical conditioning. Moreover, finite-precision effects and imperfect satisfaction of persistent excitation (rank loss or near-rank-deficiency of $\Psi(\theta_n)$) may prevent attaining $\tilde J(\hat\theta_n)=0$ even when the structural model class is correct. In such cases, the conclusion of Theorem \ref{thm:MII_vp convergence} should be interpreted as a local recovery statement. When the optimization is initialized sufficiently close  and the problem is well-conditioned, the algorithm recovers the true parameters (up to similarity), but global recovery is not guaranteed.
\end{rmk}

 \subsection{Computational cost analysis} \label{sec: Computational costs}

This subsection presents the computational costs in terms of flops for the steps of MIIvp and other practical considerations.  MII basically consists of three SVDs in a row and the rest is simple matrix operations.  Consider the case of a single input and  output dimension $y_k \in \mathbb{R}^{n_y}$. The costs of each step are summarized in Table~\ref{tab:comp_costs}.

The Pencil method involves the solution of a linear least squares problem and the SVD of the Hankel matrix. The solution of the linear least squares~\eqref{OLS_pencil} costs 
$\mathcal{O}(m^3+m^{2}n_y)$ \citep{Golub_Matrix_computations}. 
The SVD of $\mathcal{H} \in \mathbb{R}^{pn_y \times p}$ with $p = m/2$ amounts to $\mathcal{O}(\frac{m^3}{8}\,n_y)$. It is also required to store the Hankel matrix in 
order to derive the linear matrices. If the number of samples and the output dimension are very high, one may use the CUR factorization instead of the SVD as in~\cite{CUR_ERA}, or a randomized SVD as in~\cite{R-svd}.

The Loewner framework, for a choice of $l$ frequency samples, requires $\mathcal{O}(\frac{n_y}{2}\,l^2)$ for the formation of the Loewner matrices and $\mathcal{O}(\frac{l^3}{8})$ for the SVD.

The inference step for the general nonlinear 
structure \eqref{struct: General_nonlinear} with dimension $r_f$ for the nonlinear functions $f(x_k, u_k)$ and $g(x_k, u_k)$ takes $\mathcal{O}\!\big(m(r{+}1{+}r_f)^2 + m(r{+}1{+}r_f)(r{+}n_y)\big)$ with $m \ge r + 1 + r_f$. Note that we also leverage the fact that the state dimension $r$ is low.

For the VarPro step, one  solves a nonlinear 
optimization of dimension $\tilde{n} = r^2 + 2r + r_f r$. For the bilinear case, this results in $\tilde{n} = 2(r^2 + r)$ and for the quadratic-bilinear case, $\tilde{n} = 2r^2 + 2r + r^3$. Crucially, $\tilde{n}$ is independent of the output dimension~$n_y$. Since L-BFGS can handle a large number of parameters, the bilinear case can consider relatively larger reduced orders (i.e.\ $r > 100$). 
The quadratic-bilinear case becomes more involved for larger orders of $r$ due to the cubic growth of $\tilde{n}$. More sophisticated numerical techniques that can speed up optimization are left as a topic of 
future studies. The second part of VarPro consists of solving a linear least squares similar to the inference step, but without the state equation, so it results in $\mathcal{O}\!\big(m(r{+}1{+}r_f)^2 + m(r{+}1{+}r_f)n_y\big)$.

\hspace{1cm}

\begin{table}[ht!]
\centering
\resizebox{\textwidth}{!}{%
\begin{tabular}{lll}
\toprule
\textbf{Step} & \textbf{Task} & \textbf{Flops} \\
\midrule
\multirow{2}{*}{Pencil}
  & Least squares \eqref{OLS_pencil} & $\mathcal{O}(m^3 + m^2 n_y)$ \\
  & SVD of $\mathcal{H}$ & $\mathcal{O}\!\left(\frac{m^3}{8}\, n_y\right)$ \\
\midrule
\multirow{2}{*}{LF}
  & Matrix assembly & $\mathcal{O}\!\left(\frac{n_y}{2}\, l^2\right)$ \\[4pt]
  & SVD of Loewner pencil & $\mathcal{O}\!\left(\frac{l^3}{8}\right)$ \\
\midrule
Inference
  & Least squares \eqref{lsqw}  & $\mathcal{O}\!\big(m(r{+}1{+}r_f)^2 + m(r{+}1{+}r_f)(r{+}n_y)\big)$ \\
\midrule
\multirow{2}{*}{VarPro}
  & Nonlinear optimization of dimension $\tilde{n}$ &  \\
%\midrule
%VarPro
  & Least squares & $\mathcal{O}\!\big(m(r{+}1{+}r_f)^2 + m(r{+}1{+}r_f)n_y\big)$ \\
\bottomrule
\end{tabular}}

\caption{Computational costs of each step of MIIvp for a single-input system with output dimension $n_y$ and $m$ samples, where $l$ is the number of frequency samples, $r$ is the reduced order, and $r_f$ is the dimension of $f(x_k, u_k)$. The flops for the nonlinear optimization step are omitted as they depend on the specific solver and details of each implementation. }
\label{tab:comp_costs}
\end{table}

\section{Numerical experiments}\label{sec: Numerics}

In this section, we present numerical experiments to evaluate the performance of  MIIvp. All computations are done on an Intel\textregistered \hspace{0.01cm}  Core\texttrademark i5-1345U   1.60 GHz, 16GB RAM,  with MATLAB\textregistered (2019b). For the implementation of MII, the MDSPACK library is used\citep{mdspack}\footnote{\url{https://mordigitalsystems.fr/en/index.html}}. For the implementation of  L-BFGS, the minFunc toolbox \citep{Schmidt_minFunc_2005} is used. We run the nonlinear optimization solver for up to 200 iterations. Increasing the number of  iterations could lead to a performance  improvement but we aim at  emphasizing the computational efficiency and keeping the overall framework fast. We denote the bilinear models  constructed by MII and MIIvp as MII$_{B}$ and MIIvp$_{B}$, respectively. The quadratic-bilinear models constructed by MII and MIIvp are denoted  as MII$_{QB}$ and MIIvp$_{QB}$, respectively. The error metrics that we consider are the Relative Frobenius Norm Error (RFNE) defined as 
\begin{equation} \label{eq:relFrob_error}
    \mathrm{RFNE}=\dfrac{\|Y-\hat{Y}\|_{F}} {\| Y\|_{F}},
\end{equation}
for the general case $n_y > 1$,  and the Root Mean Squared Error (RMSE)  defined as (for the case $n_y=1$):
\begin{equation}
\mathrm{RMSE}
= \sqrt{\frac{1}{m}\sum_{i=0}^{m-1}\left(y_i-\hat{y}_i\right)^2}.
\end{equation}

\subsection{Application on a simple model} \label{subsec:Toy_model}

For the first numerical experiment we consider  a bilinear model   with zero initial condition $x_0$. The model has state and output dimensions $n=2, n_y=100$, respectively, and it consists of the  following matrices

$$
A=
\begin{bmatrix}
0.7 & 0.1\\
-0.2 & 0.9
\end{bmatrix},
~~
B=
\begin{bmatrix}
0.5\\
0.1
\end{bmatrix},
~~ 
N=
\begin{bmatrix}
0.05 & -0.02\\
0.01 & 0.03
\end{bmatrix}
$$
$$
C_{ij}= \sin\!\left(\frac{j\,\pi\,(i-1)}{n_y-1}\right), 
~~ 
F_{ij}= 0.1\,\cos\!\left(\frac{j\,\pi\,(i-1)}{n_y-1}\right),
$$
%\end{align*}
%\hspace{0.1cm}

for $i=1,\ldots,n_y$ and $j=1,2$, and $D=0$. The matrices $C$ and $F$ are defined as trigonometric functions to promote a well-conditioned output equation.
The eigenvalues of the $A$ matrix  are  $\{ 0.8 \pm 0.1\imath\}$ and those of the  $N$ matrix are  $\{0.04 \pm 0.01\imath\}$. We consider the input 
\begin{equation} \label{eq:u_train}
    u_t=0.8\sin\!\bigl(2\pi\cdot 0.02\, t\bigr)+0.4\sin\!\bigl(2\pi\cdot 0.05\, t+0.3\bigr)
\end{equation}

for  $t=0,1,2,\ldots, m-1$ with
$m=800$   and time-step $dt=1$ sec and then we normalize it (for possibly better numerical robustness)  by  dividing by its standard deviation.

We apply the Pencil method from the collected input-output data   to construct a model of the same order $2$.  The resulting linear model has a relative Frobenius norm error  \eqref{eq:relFrob_error} of $28.55 \%$.  By applying MII (where Tikhonov regularization was used with penalty parameter $\mu=0.41$  chosen by the L-curve criterion \citep{Hansen1992})   the error  improves to $20.38\%$. We initialize the MIIvp by the model from MII  and the resulting relative error drops almost to machine precision, reaching $1.45 \times 10^{-9} \%$. Furthermore, MIIvp is able to identify the eigenvalues of the true underlying model, for both the $A$ and $N$ matrices, as shown in  Table \ref{tab:eigA_eigN}.
In this example, MIIvp converged after 190 iterations with a tolerance of $10^{-6}$  and an execution time of approximately $1.1$ sec. Further decreasing the tolerance would lead to even lower error at the cost of more iterations and execution time. Note that the dimension of the nonlinear optimization problem of  VarPro is $2(n^2+n)=12$,  whereas if we considered the case of  nonlinear optimization    over all parameters the dimension would be equal to  $512$. %$dim (\theta)=512$.

\begin{table}[hbt!]
    \centering
    \begin{tabular}{|c|c|c|}
        \hline
        Model & $eig(A)$ & $eig(N)$\\
        \hline
        True & $\{\,0.8  \pm 0.1\imath\}$ &$\{\,0.04  \pm 0.01\imath\}$\\
        \hline
        Linear & $\{\,0.8680  \pm 0.1681\imath\}$& -\\
        \hline
        MII$_{B}$  & $\{\,0.8681 \pm 0.1681\imath\}$ &$\{\,-4.92\times10^{-7},\;-1.84\times10^{-6}\,\}$\\
        \hline
        \textbf{MIIvp$_{B}$} &  $\mathbf{\{\,0.8  \pm 0.1\imath\}}$& $\mathbf{\{\,0.04  \pm 0.01\imath\}}$\\
        \hline
    \end{tabular}
    \caption{Eigenvalues of the $A$ matrix for the original, linear, MII and MIIvp models, and the eigenvalues of the $N$ matrix for the original, MII and MIIvp models.}
    \label{tab:eigA_eigN}
\end{table}

Next, we validate the models on a validation dataset with an input modified in frequency and amplitude  as
\begin{equation} \label{eq:u_val_sin}
    u^{\text{val}}_{t}= 2 \sin\!\bigl(2\pi\cdot {0.04}\, t\bigr)+0.2\sin\!\bigl(2\pi\cdot 0.08\, t+0.1\bigr),
\end{equation}

for $t=0,1,2,\cdots,m-1$ for $m=1000$. The results of all models in terms of RFNE for both training  and validation  are presented in Table \ref{tab:Val_train_sin}.
We notice that the result for MIIvp stays at the same error magnitude as the training one.

\begin{table}[hbt!]
    \centering
    \begin{tabular}{|c|c|c|}
    \hline
         Model& RFNE $(\%)$ (Training) &RFNE $(\%)$ (Validation) \\
         \hline
         Linear& $28.55$&$93.6$ \\
         \hline
         MII$_{B}$& $20.38$&$19.62$\\
         \hline
         \textbf{MIIvp$_{B}$}&$\mathbf{1.45 \times 10^{-9}}$ & $\mathbf{1.56 \times 10^{-9}}$ \\
         \hline
    \end{tabular}
    \caption{RFNE of different models for the training dataset and for the validation dataset with the input \eqref{eq:u_val_sin}.}
    \label{tab:Val_train_sin}
\end{table}

Finally, we conduct 1000 trials with a random input (with MATLAB's \emph{randn})   for $m=1000$ samples. The results for the relative Frobenius errors   are displayed in Table \ref{tab:Monte_carlo}.  Again, the order of the errors for  MIIvp stays at the same magnitude demonstrating that the method properly learned the underlying dynamics of the  model.

\begin{table}[!t]
\centering
%\resizebox{\columnwidth}{!}{%
\begin{tabular}{|c|c|c|c|}
  \hline
  Model & Mean (\%) & Min (\%) & Max (\%) \\
  \hline
  Linear  & $50.9$ & $45.83$ & $56.66$ \\
  \hline
  MII$_{B}$   & $67.49$ & $59.86$ & $75.59$ \\
  \hline
  \textbf{MIIvp$_{B}$} & $\mathbf{6.03 \times 10^{-9}}$ & $\mathbf{5.34 \times 10^{-9}}$ & $\mathbf{6.79 \times 10^{-9}}$ \\
  \hline
\end{tabular}%
%}
\caption{RFNE (\%) over 1000 random input trials.}
\label{tab:Monte_carlo}
\end{table}

\subsection{Nonlinear benchmarks}
Next, we apply MIIvp to two nonlinear  SISO benchmark datasets from \url{https://www.nonlinearbenchmark.org/benchmarks}.  We compare MIIvp with other methods reported on the aforementioned website. In each benchmark, the  comparison is made in terms of performance on unseen data  and not those that we use for training. The training and validation periods are based on the guidelines given for each benchmark on the aforementioned website.  For  comparison,  the suggested error metric is the  RMSE.

\subsubsection{Application on the Coupled Electric Drives (CED) benchmark} \label{subsec:CED}

This benchmark is called the Coupled Electric Drives (CED)  and can be found at \url{https://www.nonlinearbenchmark.org/benchmarks/coupled-electric-drives}. The CED benchmark describes a two-motor, belt-driven pulley system with lightly damped dynamics. It focuses on speed control, with angular speed measured by an unsigned sensor, while the short datasets make identification difficult. Two different experiments are considered: Test 1 $(u_1,y_1)$ and Test 2  $(u_2,y_2)$, with $m=500$ samples  for each input-output pair and  a $dt=0.02$ sec time-step. To validate the model's performance for each input-output pair, the first $400$  samples are considered for training and the rest $100$  for prediction. The order of the system is chosen to be $r=3$, as was chosen by other methods reported on the benchmark website. 
For MIIvp, better results for this benchmark are obtained by the quadratic-bilinear structure. Table \ref{tab:CED_results} presents the results of various methods for the unseen data of Test 1 and Test 2  where the MIIvp with the quadratic-bilinear structure is denoted as MIIvp$_{QB}$. In this benchmark, MIIvp has lower RMSE compared to all methods except for dynoNet \citep{Forgione2020dynoNet}. It also has significantly lower execution time compared to other methods in Table \ref{tab:CED_results}. The output comparison of the original data and the model constructed by MIIvp is displayed in Figure \ref{fig:CED_QB}. The vertical black dotted line denotes the end of training at $t=8 $  sec  ($m=400$ samples).

\begin{table}[hbt!]
    \centering
    \begin{tabular}{|c|c|c|}
    \hline
         Model& RMSE (Training)  & RMSE (Validation)\\
         \hline
         Linear& 0.5069  & 0.3704\\
         \hline
       MII$_{QB}$& $0.3363$ & $0.2859$\\
         \hline  
\textbf{MIIvp$_{QB}$}& $\mathbf{0.0494}$ &$\mathbf{0.0662}$\\
         \hline
    \end{tabular}
    \caption{Training and validation RMSE of the linear model, MII model and the one from MIIvp for the Test 1 data of the CED benchmark.}
    \label{tab:results_Y1}
\end{table}

\begin{table}[hbt!]
    \centering
    \begin{tabular}{|c|c|c|}
    \hline
         Model& RMSE (Training)  & RMSE (Validation)\\
         \hline
         Linear& 0.4856  & 0.3595\\
         \hline
       MII$_{QB}$& 0.3504 & 0.3328\\
         \hline  
\textbf{MIIvp$_{QB}$}& \textbf{0.0496} &\textbf{0.0506}\\
         \hline
    \end{tabular}
    \caption{Training and validation RMSE of the linear model, MII model and the one from MIIvp for the Test 2 data of the CED benchmark.}
    \label{tab:results_Y2}
\end{table}

\begin{table}[ht!]
\centering
\resizebox{\textwidth}{!}{
\begin{tabular}{|c|c|c|c|}
\hline
\textbf{Method} & \textbf{Test 1} & \textbf{Test 2} &\textbf{Time} \\
\hline
RBFNN - FSDE & 0.130 & 0.185& -\\
\hline
Cascaded Splines & 0.216 & 0.110&-\\
\hline
Extended Fuzzy Logic & 0.150 & 0.092& -\\
\hline
neural ODE & 0.362 & 0.370& -\\
\hline
Sparse Bayesian LSTM & 0.121 & 0.097& -\\
\hline
Sparse Bayesian MLP & 0.149 & 0.120& -\\
\hline
GP with squared exponential kernel & 0.153 & 0.132 &-\\
\hline
GP with rational quadratic kernel & 0.150 & 0.167& -\\
\hline
CT SUBNET & 0.143 & 0.100&- \\
\hline
neural ODE with normalization & 0.198 & 0.158& -\\
\hline
CT SUBNET neural state-space encoder & 0.115 & 0.074 &$1-60$ min\\
\hline
SUBNET neural state-space encoder & 0.169 & 0.117&$1-60$ min \\
\hline
dynoNet & 0.062 & 0.047&$1-60$ sec \\
\hline
ARX & 0.433 & 0.179& $ < 1$ sec\\
\hline
GPNARX & 0.0870 & 0.0696& $1-60$ min\\
\hline
GRU & 0.180 & 0.159& $1-60$ min\\
\hline
LSTM & 0.139 & 0.111& $1-60$ min\\
\hline
MLPFIR & 0.585 & 0.962 &$1-60$ sec\\
\hline
MLPNARX & 0.123 & 0.193&$1-60$ sec \\
\hline
OLSTM & 0.181 & 0.191& $1-60$ min\\
\hline
PNARX & 0.230 & 0.137& $ < 1$ sec\\
\hline
RNN & 0.114 & 0.138 &$1-60$ min\\
\hline
\textbf{MIIvp$_{QB}$ (this work)} & 0.0662 & 0.0506& $\approx 0.85$ sec  ($ < 1$ sec)\\
\hline
\end{tabular}}

\caption{Results for the CED benchmark for the test 1  and test 2 as  found at  \url{https://www.nonlinearbenchmark.org/benchmarks/coupled-electric-drives}. The (-) indicates that the execution time was not reported.}
\label{tab:CED_results}

\end{table}

\begin{figure*}[h!]
    %\centering
    \hspace*{-1.5cm}
    \includegraphics[width=1.2\linewidth]{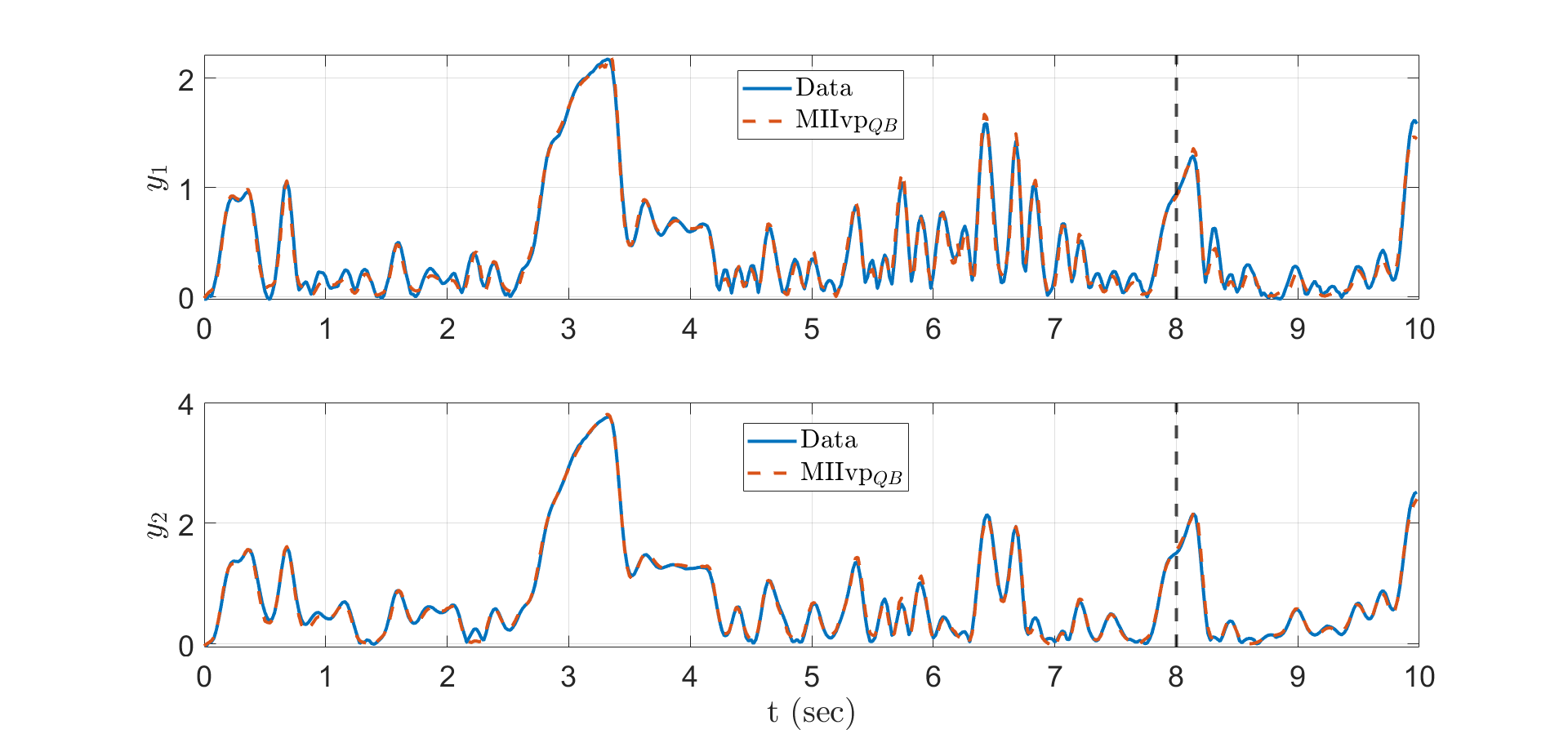}
    \caption{Comparison of the original (solid blue) and MIIvp$_{QB}$ (red dashed) outputs for Test 1 (upper panel) and Test 2 (lower panel) of the CED benchmark. The vertical black dotted line denotes the end of training at time $t=8$ sec.}
    \label{fig:CED_QB}
\end{figure*}

\subsubsection{Application on the Cascaded Tanks  (CT) benchmark}\label{subsec:CT}

The second benchmark is the Cascaded Tanks (CT) benchmark, available at \url{https://www.nonlinearbenchmark.org/benchmarks/cascaded-tanks}. The  CT benchmark is a fluid level control system composed of two serially connected tanks fed by a pump. It exhibits both soft and hard nonlinearities, including overflow effects, and the relatively short data records make system identification challenging.
The benchmark consists of two different experiments with $m=1024$ samples and time-step $dt=4$ sec. The first input-output pair is used for training and the second input-output pair is used for  validation.
Again, the choice of order is $r=3$, as for the PNLSS and NLSS methods in \cite{RELAN2017_CTbenchmark}  and as before, the quadratic-bilinear structure   is chosen. Results of methods, in terms of performance for the validation data, are presented in Table \ref{tab:CT_results_methods}. The comparison of the validation output and the one from MIIvp is shown in Figure \ref{fig:Ct_QB}.  For this benchmark, MIIvp produced a higher RMSE  for the validation dataset. In terms of the error metrics, it outperformed 7 out of the 24 reported methods. These results may indicate that quadratic and bilinear terms are not sufficient to model the nonlinear behavior of this physical process as methods that use higher polynomial terms (such as PNLSS) and especially neural networks (such as NLSS, dynoNet, SUBNET)  seem to be more appropriate for this benchmark. Since, for  the estimation and validation the same initial condition for the state was used, it could be the case that MIIvp underestimated it as Figure \ref{fig:Ct_QB} depicts. However, MIIvp still produced a stable (for the operating regimes) reduced-order model relatively fast ($\approx 1.8$ sec) and with a structure that is easy to interpret.   
\begin{table}[ht!]
    \centering
    \begin{tabular}{|c|c|c|}
    \hline
         Model& RMSE (Training)  & RMSE (Validation)\\
         \hline
         Linear& 1.16  & 1.4784\\
         \hline
       MII$_{QB}$& 0.6406 & 1.0562\\
         \hline  
\textbf{MIIvp$_{QB}$}& \textbf{0.1852} &\textbf{0.5121}\\
         \hline
\end{tabular}
    \caption{Training and validation RMSE of the linear model, MII model and  MIIvp model for the dataset  of the CT benchmark.}
    \label{tab:results_CT_mii}
\end{table}

\begin{table}[ht!]
\centering
\resizebox{\textwidth}{!}{
\begin{tabular}{|c|c|c|}
\hline
\textbf{Method} & \textbf{RMSE (Val)}& \textbf{Time} \\
\hline
Grey-Box with physical overflow model & 0.18&- \\
\hline
State-space with GP-inspired prior & 0.45 &-\\
\hline
NL-SS + NLSS2 & 0.34 &-\\
\hline
Polynomial Nonlinear State-Space & 0.45&- \\
\hline
Volterra & 0.53 &-\\
\hline
Tensor B-splines & 0.30& -\\
\hline
Truncated Simulation Error Minimization & 0.33& -\\
\hline
Soft-Constrained Integration & 0.40 &-\\
\hline
IO stable CT ANN & 0.39 &-\\
\hline
neural ODE with normalization & 0.33&- \\
\hline
CT SUBNET & 0.30 &-\\
\hline
CT SUBNET neural state-space encoder & 0.306&$1-60$ min \\
\hline
SUBNET neural state-space encoder & 0.37 &$1-60$ min\\
\hline
dynoNet & 0.421 &$1-60$ sec\\
\hline
ARX & 0.697 &$1-60$ sec\\
\hline
GPNARX & 0.54& $1-60$ min\\
\hline
GRU & 0.568 &$1-60$ min\\
\hline
LSTM & 0.452 &$1-60$ min\\
\hline
MLPFIR & 1.24 &$1-60$ min\\
\hline
MLPNARX & 1.29 &$1-60$ sec\\
\hline
OLSTM & 0.427 &$1-60$ min\\
\hline
PNARX & 0.416 &$1-60$ sec\\
\hline
RNN & 0.97 &$1-60$ min\\
\hline
Normalized Gray-box State Space Neural Networks & 0.221 &$1-60$ min\\
\hline
%MIIvp (this work) & 0.413 \\
\textbf{MIIvp$_{QB}$ (this work)} & 0.5121& $\approx 1.8$ sec($1-60$ sec) \\
\hline
\end{tabular}}
\caption{Results for the CT benchmark validation  set as  found at \url{https://www.nonlinearbenchmark.org/benchmarks/cascaded-tanks}. The (-) indicates that the execution time was not reported.}
\label{tab:CT_results_methods}
\end{table}
\begin{figure*}[ht!]
        %\centering
\hspace*{-1.7cm}
\includegraphics[width=1.2\linewidth]{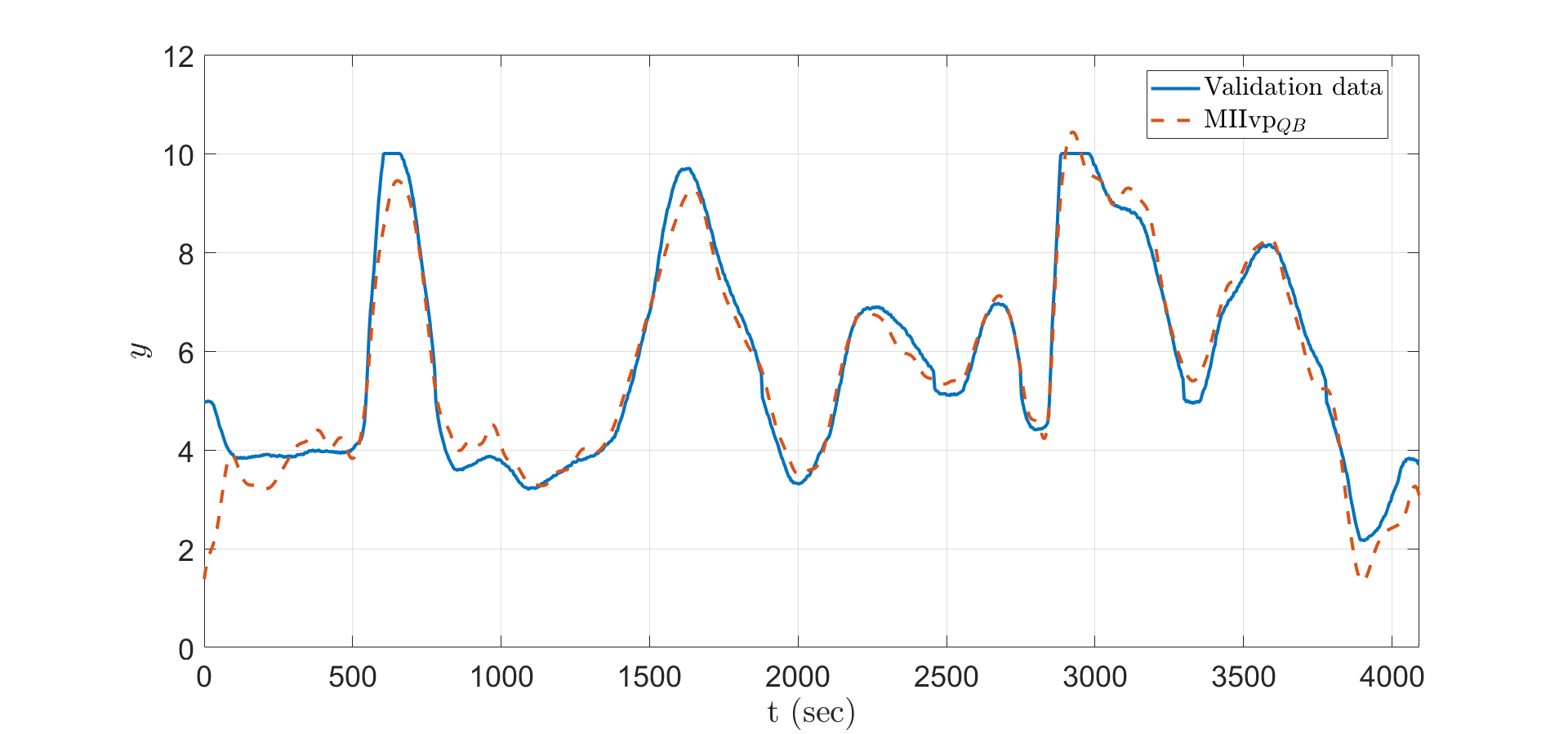}
        \caption{Comparison of the validation output series against the output of the quadratic-bilinear model constructed from MIIvp.}
        \label{fig:Ct_QB}
    \end{figure*}

\section{Conclusion} \label{sec: Conclusion}
In this work, a method capable of constructing nonlinear models of the form \eqref{struct: General_nonlinear} from input-output data in the time-domain is introduced. To  capture the nonlinear behavior of the data, the VarPro algorithm is used. The synthetic bilinear example demonstrated that MIIvp was able to recover the true model. For the more challenging nonlinear benchmarks, MIIvp achieved second-lowest reported validation error on both tests of the CED benchmark. For the CT benchmark, the suggested quadratic-bilinear structure seemed inadequate to model the nonlinearities compared to other choices such as different types of neural networks. Among the advantages of MIIvp is its flexibility since any nonlinear function may be considered, as long as it is linear in the parameters. A disadvantage is that the stability of the constructed model cannot be guaranteed. This is a common issue for many non-intrusive methods and a possible way to tackle this is  to add constraints on the nonlinear optimization  problem and solve it with L-BFGS with box constraints (L-BFGS-B) as in \cite{Bemporad2025}. Moreover, different nonlinear solvers may be considered, which could possibly further improve the method's performance depending on the nature of the problem. Another important topic to investigate is the effect of noise on the proposed method.  Future work will aim to circumvent these issues and try to improve  the optimization landscape.  Finally, the results obtained for a complex case of air pollution simulations  (where the output dimension is very large (i.e. $n_y > 10^5$) will be included  in a future study.

\bibliographystyle{apalike}
\bibliography{cas-refs}

@article{AGPV25,
    title={{The Loewner framework for parametric systems: Taming the curse of dimensionality}}, 
    author={A.~C. Antoulas and I~-V. Gosea and C. Poussot-Vassal},
    year={2025},
    Journal={SIAM Review},
    Volume = {67},
    Number = {4},
    Pages = {737-770},
    note ={\url{https://doi.org/10.1137/24M1656657}}, 
}

@misc{mdspack,
	author = {{MOR Digital Systems}},
  	title = {{MDSPACK (v1.1.0)}},
  	note = {Main page (\url{https://mordigitalsystems.fr}) \& Documentation (\url{https://mordigitalsystems.fr/static/mdspack_html/MDSpack-guide.html})},
  	year = {2025},
}

@article{Golub1973,
author = {Golub, G. H. and Pereyra, V.},
title = {The Differentiation of Pseudo-Inverses and Nonlinear Least Squares Problems Whose Variables Separate},
journal = {SIAM Journal on Numerical Analysis},
volume = {10},
number = {2},
pages = {413-432},
year = {1973},
doi = {10.1137/0710036},
}

@article{MixedII,
  title={Mixed interpolatory and inference non-intrusive reduced order modeling with application to pollutants dispersion},
  author={Charles Poussot-Vassal and Tiphaine Sabatier and Claire Sarrat and Pierre Vuillemin},
  journal={arXiv:2012.07126},
  year={2020},
  }

@article{Hansen1992,
  author  = {Hansen, Per Christian},
  title   = {Analysis of Discrete Ill-Posed Problems by Means of the {L}-Curve},
  journal = {SIAM Review},
  volume  = {34},
  number  = {4},
  pages   = {561--580},
  year    = {1992}
}

@article{RELAN2017_CTbenchmark,
title = {An Unstructured Flexible Nonlinear Model for the Cascaded Water-tanks Benchmark},
journal = {IFAC-PapersOnLine},
volume = {50},
number = {1},
pages = {452-457},
year = {2017},
note = {20th IFAC World Congress},
issn = {2405-8963},
doi = {https://doi.org/10.1016/j.ifacol.2017.08.074},
url = {https://www.sciencedirect.com/science/article/pii/S2405896317300940},
author = {Rishi Relan and Koen Tiels and Anna Marconato and Johan Schoukens},
}

@article{CUR_ERA,
author = {Kramer, Boris and Gorodetsky, Alex A.},
title = {System Identification via CUR-Factored Hankel Approximation},
journal = {SIAM Journal on Scientific Computing},
volume = {40},
number = {2},
pages = {A848-A866},
year = {2018},
doi = {10.1137/17M1137632},
URL = { https://doi.org/10.1137/17M1137632},
}

@article{R-svd,
author = {Halko, N. and Martinsson, P. G. and Tropp, J. A.},
title = {Finding Structure with Randomness: Probabilistic Algorithms for Constructing Approximate Matrix Decompositions},
journal = {SIAM Review},
volume = {53},
number = {2},
pages = {217-288},
year = {2011},
doi = {10.1137/090771806},
}

@book{Golub_Matrix_computations,
author = {Golub, Gene H. and Van Loan, Charles F.},
title = {Matrix Computations - 4th Edition},
publisher = {Johns Hopkins University Press},
year = {2013},
doi = {10.1137/1.9781421407944},
address = {Philadelphia, PA},
edition   = {},
URL = {https://epubs.siam.org/doi/abs/10.1137/1.9781421407944},
eprint = {https://epubs.siam.org/doi/pdf/10.1137/1.9781421407944}
}

@misc{Schmidt_minFunc_2005,
  author       = {Schmidt, Mark},
  title        = {minFunc: Unconstrained Differentiable Multivariate Optimization in Matlab},
  year         = {2005},
  note         = {Matlab software package},
  howpublished = {Available at: \url{https://www.cs.ubc.ca/~schmidtm/Software/minFunc.html}},
  url          = {https://www.cs.ubc.ca/~schmidtm/Software/minFunc.html},
  
}

@article{ERA1985,
author = {Juang, Joe and Pappa, Richard},
year = {1985},
month = {11},
pages = {},
title = {An Eigensystem Realization Algorithm for Modal Parameter Identification and Model Reduction},
volume = {8},
journal = {Journal of Guidance Control and Dynamics},
doi = {10.2514/3.20031},
}

@INPROCEEDINGS{Enforcing_stability,
  author={Gosea, Ion Victor and Poussot-Vassal, Charles and Antoulas, Athanasios C.},
  booktitle={2021 29th Mediterranean Conference on Control and Automation (MED)}, 
  title={On enforcing stability for data-driven reduced-order models}, 
  year={2021},
  volume={},
  number={},
  pages={487-493},
  doi={10.1109/MED51440.2021.9480216}}

@book{Pintelon_Schoukens_book2012,
title = "System Identification: A frequency Domain Approach",
author = "Rik Pintelon and Joannes Schoukens",
note = "IEEE Press, Wiley",
year = "2012",
month = apr,
day = "1",
language = "English",
isbn = "978-0-470-64037-1",
publisher = "Wiley / IEEE Press",
edition = "2",
}

@article{Loewner,
    author = {A. J. Mayo and A. C. Antoulas},
    title = {A framework for the solution of the generalized realization problem},
    journal = {Model Reduction of Complex Dynamical Systems. International Series of Numerical Mathematics, vol 171.},
    year = {2007}
}

@inproceedings{MyCDC,
  author={Xylogiannis, Dimitrios and Poussot-Vassal, Charles and Sarrat, Claire},
  booktitle={2024 IEEE 63rd Conference on Decision and Control (CDC)}, 
  title={Mixed Interpolatory Inference for reduced order bilinear model construction}, 
  year={2024},
  volume={},
  number={},
  pages={8028-8033},
  doi={10.1109/CDC56724.2024.10886585}}

@inproceedings{MyECC,
  author={Xylogiannis, Dimitrios and Poussot-Vassal, Charles and Sarrat, Claire},
  booktitle={2025 European Control Conference (ECC)}, 
  title={Construction of bilinear systems with nuclear norm regularization}, 
  year={2025},
  volume={},
  number={},
  pages={2703-2708},
  doi={10.23919/ECC65951.2025.11187103}}

@book{benner2021ROM_vol1,
  title     = {Model Order Reduction. Volume 1: System- and Data-Driven Methods and Algorithms},
  editor    = {Benner, Peter and Grivet-Talocia, Stefano and Quarteroni, Alfio and Rozza, Gianluigi and Schilders, Wil and Silveira, Lu{\'i}s Miguel},
  year      = {2021},
  publisher = {De Gruyter},
  address   = {Berlin, Boston},
  isbn      = {978-3-11-050043-1},
  doi       = {10.1515/9783110498967}
}

@Article{Sarrat2017airport,
AUTHOR = {Sarrat, Claire and Aubry, Sébastien and Chaboud, Thomas and Lac, Christine},
TITLE = {Modelling Airport Pollutants Dispersion at High Resolution},
JOURNAL = {Aerospace},
VOLUME = {4},
YEAR = {2017},
NUMBER = {3},
ARTICLE-NUMBER = {46},
URL = {https://www.mdpi.com/2226-4310/4/3/46},
ISSN = {2226-4310},
DOI = {10.3390/aerospace4030046}
}

@book{antoulas2005approximation,
  author    = {Antoulas, Athanasios C.},
  title     = {Approximation of Large-Scale Dynamical Systems},
  year      = {2005},
  publisher = {Society for Industrial and Applied Mathematics},
  address   = {Philadelphia, PA},
  series    = {Advances in Design and Control},
  volume    = {6},
  doi       = {10.1137/1.9780898718713},
  isbn      = {978-0-89871-529-3}
}

@article{GoseaDMD,
  author    = {I.V.  Gosea and I. Pontes-Duff},
  journal   = {arXiv:2003.06484},
  title     = {Toward fitting structured nonlinear systems by means of dynamic mode decomposition},
  year      = {2020},
  month     = {March},
  volume    = {},
  pages     = {1-16},
  number    = {},
}

@book{AntoulasBook:2020,
	Author = {A.C. Antoulas and C.A. Beattie and S. Gugercin},
	Title = {{Interpolatory methods for model reduction}},
	Address = {Philadelphia},
	Publisher = {SIAM Computational Science and Engineering},
	Year = {2020}}

@ARTICLE{Decuyper_decouple,
  author={Decuyper, Jan and Dreesen, Philippe and Schoukens, Johan and Runacres, Mark C. and Tiels, Koen},
  journal={IEEE Control Systems Letters}, 
  title={Decoupling Multivariate Polynomials for Nonlinear State-Space Models}, 
  year={2019},
  volume={3},
  number={3},
  pages={745-750},
  doi={10.1109/LCSYS.2019.2916955},
  }

@article{OptDMD2018,
author = {Askham, Travis and Kutz, J. Nathan},
title = {Variable Projection Methods for an Optimized Dynamic Mode Decomposition},
journal = {SIAM Journal on Applied Dynamical Systems},
volume = {17},
number = {1},
pages = {380-416},
year = {2018},
doi = {10.1137/M1124176},
url = {https://doi.org/10.1137/M1124176},
}

@article{BRULS1999,
title = {Linear and Non-linear System Identification Using Separable Least-Squares},
journal = {European Journal of Control},
volume = {5},
number = {1},
pages = {116-128},
year = {1999},
issn = {0947-3580},
doi = {https://doi.org/10.1016/S0947-3580(99)70146-9},
url = {https://www.sciencedirect.com/science/article/pii/S0947358099701469},
author = {J. Bruls and C.T. Chou and B.R.J. Haverkamp and M. Verhaegen},
}

@ARTICLE{GanMin2021,
  author={Gan, Min and Guan, Yu and Chen, Guang-Yong and Chen, C. L. Philip},
  journal={IEEE Transactions on Neural Networks and Learning Systems}, 
  title={Recursive Variable Projection Algorithm for a Class of Separable Nonlinear Models}, 
  year={2021},
  volume={32},
  number={11},
  pages={4971-4982},
  doi={10.1109/TNNLS.2020.3026482},}

@article{Liu1989,
  title={On the limited memory BFGS method for large scale optimization},
  author={Dong C. Liu and Jorge Nocedal},
  journal={Mathematical Programming},
  year={1989},
  volume={45},
  pages={503-528},
  url={https://api.semanticscholar.org/CorpusID:5681609}
}

@article{Bemporad2025,
author={Bemporad, Alberto},
  journal={IEEE Transactions on Automatic Control}, 
  title={An {L-BFGS-B} Approach for Linear and Nonlinear System Identification Under $\ell _{1}$ and Group-Lasso Regularization}, 
  year={2025},
  volume={70},
  number={7},
  pages={4857-4864},
  doi={10.1109/TAC.2025.3541018}}

@InProceedings{Beintema21a,
  title = 	 {Nonlinear state-space identification using deep encoder networks},
  author =       {Beintema, Gerben and Toth, Roland and Schoukens, Maarten},
  booktitle = 	 {Proceedings of the 3rd Conference on Learning for Dynamics and Control},
  pages = 	 {241--250},
  year = 	 {2021},
  volume = 	 {144},
  series = 	 {Proceedings of Machine Learning Research},
  month = 	 {07 -- 08 June},
  publisher =    {PMLR},
  url = 	 {https://proceedings.mlr.press/v144/beintema21a.html},
  }

@article{BEINTEMA2021_video,
title = {Non-linear State-space Model Identification from Video Data using Deep Encoders},
journal = {IFAC-PapersOnLine},
volume = {54},
number = {7},
pages = {697-701},
year = {2021},
note = {19th IFAC Symposium on System Identification SYSID 2021},
issn = {2405-8963},
doi = {https://doi.org/10.1016/j.ifacol.2021.08.442},
url = {https://www.sciencedirect.com/science/article/pii/S2405896321012167},
author = {Gerben I. Beintema and Roland Toth and Maarten Schoukens},
}

@article{WAHLSTROM2015,
title = {Learning deep dynamical models from image pixels},
journal = {IFAC-PapersOnLine},
volume = {48},
number = {28},
pages = {1059-1064},
year = {2015},
note = {17th IFAC Symposium on System Identification SYSID 2015},
issn = {2405-8963},
doi = {https://doi.org/10.1016/j.ifacol.2015.12.271},
url = {https://www.sciencedirect.com/science/article/pii/S2405896315028955},
author = {Niklas Wahlström and Thomas B Schön and Marc Peter Deisenroth},
}

@article{FORGIONE2021,
title = {Continuous-time system identification with neural networks: Model structures and fitting criteria},
journal = {European Journal of Control},
volume = {59},
pages = {69-81},
year = {2021},
issn = {0947-3580},
doi = {https://doi.org/10.1016/j.ejcon.2021.01.008},
url = {https://www.sciencedirect.com/science/article/pii/S0947358021000169},
author = {Marco Forgione and Dario Piga},
}

@article{Forgione2020dynoNet,
  title={dynoNet: A neural network architecture for learning dynamical systems},
  author={Marco Forgione and Dario Piga},
  journal={International Journal of Adaptive Control and Signal Processing},
  year={2020},
  volume={35},
  pages={612 - 626},
  url={https://api.semanticscholar.org/CorpusID:219260808}
}

@article{BEINTEMA2023,
title = {Deep subspace encoders for nonlinear system identification},
journal = {Automatica},
volume = {156},
pages = {111210},
year = {2023},
issn = {0005-1098},
doi = {https://doi.org/10.1016/j.automatica.2023.111210},
url = {https://www.sciencedirect.com/science/article/pii/S0005109823003710},
author = {Gerben I. Beintema and Maarten Schoukens and Roland Tóth},
}

@article{Verdult2001,
title = "Identification of multivariable bilinear state space systems based on subspace techniques and separable least squares optimization",
author = "V. Verdult and M. Verhaegen",
year = "2001",
language = "Undefined/Unknown",
volume = "74",
pages = "1824--1836",
journal = "International Journal of Control",
issn = "0020-7179",
publisher = "Taylor & Francis",
number = "18",
}

@article{Ho_Kalman1966,
    author ={B.L Ho and R.E. Kalman} ,
    title = {Effective construction of linear state variable models from input-output functions},
    journal = {Automatisierungstechnik},
    doi = {doi:10.1524/auto.1966.14.112.545},
    year = {1966}
}

@article{Marconato2014ImprovedIF,
  title={Improved Initialization for Nonlinear State-Space Modeling},
  author={Anna Marconato and Jonas Sj{\"o}berg and Johan A. K. Suykens and Johan Schoukens},
  journal={IEEE Transactions on Instrumentation and Measurement},
  year={2014},
  volume={63},
  pages={972-980},
  url={https://api.semanticscholar.org/CorpusID:825462}
}

@article{PADUART2010_pnlss,
title = {Identification of nonlinear systems using Polynomial Nonlinear State Space models},
journal = {Automatica},
volume = {46},
number = {4},
pages = {647-656},
year = {2010},
issn = {0005-1098},
doi = {https://doi.org/10.1016/j.automatica.2010.01.001},
url = {https://www.sciencedirect.com/science/article/pii/S000510981000021X},
author = {Johan Paduart and Lieve Lauwers and Jan Swevers and Kris Smolders and Johan Schoukens and Rik Pintelon},
}

@article{Schoukens:2016,
	Author = {J. Schoukens and  M. Vaes and R. Pintelon},
	Issue = {},
	Month = {June},
	Journal = {IEEE Control Systems Magazine},
	Pages = {38-88},
	Title = {{Linear System Identification in a Nonlinear Setting}},
	Volume = {},
	Year = {2016},
}

@inbook{AntoulasSurvey:2016,
	Author = {A.C. Antoulas and S. Lefteriu and A.C. Ionita},
	Chapter = {{A tutorial introduction to the Loewner framework for model reduction}},
	Publisher = {P. Benner, A. Cohen, M. Ohlberger and K. Willcox Eds},
	Series = {SIAM, Philadelphia},
	Title = {{Model reduction and approximation theory and algorithms}},
	Year = {2016}}

@article{IonitaPencil:2012,
   	author    = {A.C. Ionita and A.C. Antoulas},
	
	title   = {{Chap.9: Matrix pencils in time and frequency domain system identification}},
	Journal = {Developments in Control Theory Towards Glocal Control, Control, Robotics and Sensors, Institution of Engineering and Technology},
	year      = {2012},
	pages     = {79–88},
}

\end{document}